\documentclass[11pt,a4paper]{amsart}
\usepackage{xcolor}
\definecolor{mg}{rgb}  {0.85, 0.,  0.85}

\usepackage[backend=biber,bibencoding=utf8,citestyle=alphabetic,bibstyle=alphabetic,isbn=false,doi=false,url=false,hyperref=true,backref=true,backrefstyle=three]{biblatex}
\renewbibmacro*{volume+number+eid}{%
	\textbf{\printfield{volume}}%
	\setunit*{\adddot}
	\setunit*{\addnbspace}
	\printfield{number}%
	\setunit{\addcomma\space}%
	\printfield{eid}}
\renewbibmacro{in:}{%
	\ifentrytype{article}{}{\printtext{\bibstring{in}\intitlepunct}}}
\DeclareFieldFormat[article]{number}{\mkbibparens{#1}}
\usepackage[UKenglish]{babel}
\usepackage[colorlinks]{hyperref}
\usepackage{amsmath,amssymb,amsthm}
\usepackage{mathrsfs}
\usepackage{verbatim}
\usepackage[utf8]{inputenc}
\usepackage[T1]{fontenc}
\usepackage{lmodern}
\usepackage{graphicx}
\usepackage{caption}
\usepackage{subcaption}
\usepackage{wrapfig}
\usepackage{microtype}
\usepackage{bm}
\usepackage{dsfont}
\usepackage[normalem]{ulem} 
\usepackage[retainorgcmds]{IEEEtrantools}
\usepackage[showframe=false]{geometry}
\usepackage[capitalise,nameinlink]{cleveref}
\usepackage{tikz-cd}
\usetikzlibrary{arrows}
\usetikzlibrary{positioning}

\usepackage{csquotes}
\usepackage{mathtools}
\usepackage{soul}
\usepackage{xfrac}
\allowdisplaybreaks
\crefformat{equation}{(#2#1#3)}
\crefrangeformat{equation}{(#3#1#4) to~(#5#2#6)}
\crefmultiformat{equation}{(#2#1#3)}%
{ and~(#2#1#3)}{, (#2#1#3)}{ and~(#2#1#3)}
\usepackage[norefs]{refcheck}   
\makeatletter
\newcommand{\refcheckize}[1]{%
  \expandafter\let\csname @@\string#1\endcsname#1%
  \expandafter\DeclareRobustCommand\csname relax\string#1\endcsname[1]{%
    \csname @@\string#1\endcsname{##1}\@for\@temp:=##1\do{\wrtusdrf{\@temp}\wrtusdrf{{\@temp}}}}%
  \expandafter\let\expandafter#1\csname relax\string#1\endcsname
}
\newcommand{\refcheckizetwo}[1]{%
  \expandafter\let\csname @@\string#1\endcsname#1%
  \expandafter\DeclareRobustCommand\csname relax\string#1\endcsname[2]{%
    \csname @@\string#1\endcsname{##1}{##2}\wrtusdrf{##1}\wrtusdrf{{##1}}\wrtusdrf{##2}\wrtusdrf{{##2}}}%
  \expandafter\let\expandafter#1\csname relax\string#1\endcsname
}
\makeatother

\refcheckize{\cref}
\refcheckize{\Cref}
\refcheckizetwo{\crefrange}
\refcheckizetwo{\Crefrange}

\newcommand\ie{{\em i.e.}~}

\def\N{\mathbb{N}}
\def\T{\mathbb{T}}

\def\Z{\mathbb{Z}}
\def\C{\mathbb{C}}

\def\F{\mathscr F}

\def\F{\mathcal F}

\def\R{\mathbb{R}}

\def\S{\mathcal S}

\def\Im{\rm Im}

\def\d{\mathrm d}

\def\({\left(}
\def\[{\left[}
\def\){\right)}
\def\]{\right]}
\def\l|{\left\lvert}
\def\r|{\right\rvert}

\def\<{\langle}
\def\>{\rangle}

\DeclareMathOperator{\Ker}{Ker}

\DeclareMathOperator{\sgn}{sgn}

\newtheorem{Theorem}{Theorem}[section]
\newtheorem{Remark}[Theorem]{Remark}
\newtheorem{Lemma}[Theorem]{Lemma}

\newtheorem{Corollary}[Theorem]{Corollary}
\newtheorem{Proposition}[Theorem]{Proposition}
\newtheorem{Definition}[Theorem]{Definition}
\newtheorem{Example}[Theorem]{Example}

\DeclareMathOperator{\Ran}{Ran}
\crefname{Lemma}{Lemma}{Lemmas}

\newcommand{\bel}{\begin{equation} \label}
\newcommand{\ee}{\end{equation}}

\newcommand{\ba}{\begin{array}}
\newcommand{\ea}{\end{array}}

\begin{document}

\title{Continuum limit for a discrete Hodge-Dirac operator on square lattices.}

\author[P.\ Miranda]{Pablo Miranda}\address{Departamento de Matem\'atica y Ciencia de la Computaci\'on, Universidad de Santiago de Chile, Las Sophoras 173. Santiago, Chile.}\email{pablo.miranda.r@usach.cl}
\author[D.\ Parra]{Daniel Parra}
\address{Departamento de Matem\'atica y Ciencia de la Computaci\'on, Universidad de Santiago de Chile, Las Sophoras 173. Santiago, Chile.}
\email{daniel.parra.v@usach.cl}

\begin{abstract}

We study the continuum limit for Dirac-Hodge operators defined on the $n$ dimensional square lattice $h\mathbb{Z}^n$ as $h$ goes to $0$. This result extends to a first order discrete differential operator the known convergence of discrete Schrödinger operators to their continuous counterpart. To be able to define such a discrete analog, we start by defining an alternative framework for a higher--dimensional discrete differential calculus. We believe that this framework, that generalize the standard one defined on simplicial complexes, could be of independent interest. We then express our operator as a differential operator acting on discrete forms to finally be able to show the limit to the continuous Dirac-Hodge operator. 

\bigskip


\bigskip

\end{abstract}
\maketitle \vspace{-1cm}

\section{Introduction and Main Result}\label{sec:intro}

The aim of this paper is to study the continuum limit for a first order \emph{discrete differential operator} on $h\Z^n$ as the mesh size $h$ goes to $0$. Such type of results have attracted growing attention following \cite{NT21} in which Nakamura and Tadano showed that discrete Schrödinger operators of the form $-\Delta+ V$ on $L^2(h\Z^d)$ converge to the corresponding Schrödinger operator $-\Delta+ V$ on $L^2(\R^d)$ in the \emph{norm resolvent sense}. Such kind of convergence, sometimes called generalized convergence to emphasize that operators are defined in different Hilbert spaces (see \cite{PZ22x} and references therein), is useful for studiying manifolds representing networks of thin tubes in the limiting process as they shrink to its underlying graph \cite{Po12}.

In the aftermath of the aforementioned result, some related techniques have been applied to the study of Fourier Multipliers \cite{CGJ21} and of Laplacians on the half-space \cite{CGJ22xb}, and for quantum graphs on the Euclidean space \cite{ENT22}. In contrast with these positive and natural results, the study of discrete counterparts of Dirac operators has yielded different results. In \cite{SU21x} the authors showed that a discrete Dirac operator on $\ell^2(h\Z^2)^2$, written in terms of first--order discrete forward or backward differences, converge only in the \emph{strong resolvent sense}. This was then confirmed in \cite{CGJ22x} where they showed that in order to obtain the norm resolvent convergence a diagonal term of order $2$ (\ie a Laplacian) is needed to be added (this is related to the  so called \emph{fermion doubling}. Our operator  does not present  that phenomenon. See \cref{lemmaseccion3}). 

Motivated by this asymmetry between the discrete and continuous setting, the strategy of this manuscript is different. Following the study of the Gauss-Bonnet operator for graphs initiated in \cite{CT15} and continued in \cite{ABDE21,Pa17,MPR23} we adopt the point view of considering on $\Z^n$ a discrete differential calculus. First, this gives rise to a \emph{discrete exterior derivative} $d_h$ acting on the Hilbert space of square--integrable cochains $\ell^2(X(h\Z^n))$ (see \cref{subsec:defzn} for a precise definition), and to the corresponding bounded self-adjoint Gauss-Bonnet operator $d_h+d_h^*$ (we refer to \cref{sec:CDS} for more detailed presentation). Afterwrds, we show that this is equivalent to considering the exterior derivative and the Hodge-Dirac operator acting on discrete differential forms in $\Z^n$ (see \cref{sec:hodgedirac}). By construction, this operator exhibits supersymmetry (as we discuss at the end of \cref{subsec:abstract}) and is of Dirac-type in the sense that $(d_h+d_h^*)^2$ corresponds to the discrete Hodge-Laplacian. 

One of the first obstacles to defining our Gauss-Bonnet operator is that it usually involves the use of the setting of \emph{simplicial complexes}. However, as one can easily attest visually, $\Z^n$ does not have simplices of dimension $2$ or greater (see \cite{Li20} for a general overview in that direction) so we need to start by presenting an alternative to this setting that allows us to encompass the structure of $\Z^n$ as well as to make sense of the varying of the mesh size $h$.

Let $\d$ be the standard boundary operator in $\R^n$ given by the exterior derivative (see \cref{limitoperator} for some properties of this operator). It follows that $\d+\d^*$ is unitarily equivalent to an operator acting in $\bigoplus_{j=0}^nL^2(\R^n;\C^{\binom{n}j})$. 
The aim of this article is to construct an embedding $T_h$ of $\ell^2(X(h\Z^n))$ into $\bigoplus_{j=0}^nL^2(\R^n;\C^{\binom{n}j})$, and prove that $T_h(d_h+d_h^*)T_h^*$ converges in the norm resolvent sense to $\d+\d^*$. More specifically,  the main result of this article is the following.
\begin{Theorem}\label{teo:main}
Let $n\in \N$. Then, for every $h>0$ there exists an embedding $$T_h:\ell^2(X(h\Z^n))\to \bigoplus_{j=0}^nL^2\Big(\R^n;\C^{\binom{n}j}\Big)$$ such that for every $z\in\C$ satisfying $\Im (z)\neq 0$ we have
\begin{equation*}
||T_h(d_h+d_h^*-z)^{-1}T_h^*-(\d+\d^*-z)^{-1}||=O(h) \text{ as }h\rightarrow 0\, .
\end{equation*}
\end{Theorem}

Let us now briefly discuss the organization of this paper. In \cref{sec:CDS} we introduce an abstract definition of a discrete differential calculus that should encompass the simplicial complex setting as well as giving place to higher order cochains in $\Z^n$. We then go on to show that it gives rise to a Hodge Theory. We believe that this definition could be of independent interest to study higher order discrete differentials forms. Finally, we demonstrate how this setting can be used in the study of $\Z^n$ and briefly discuss why the case $n=1$ plays a particular role.

In \cref{sec:hodgedirac} we construct an intermediary space to show that the Gauss-Bonnet operator defined in \cref{sec:CDS} is unitarily equivalent to a Hodge-Dirac operator acting on \emph{discrete differential forms} in $\Z^n$. This equivalence allows us to exploit the machinery of differential forms which, together with the fact that we have global coordinates, much enhances our ability to study the corresponding operator. In particular, we show why it is natural to take $ \bigoplus_{j=0}^nL^2\Big(\R^n;\C^{\binom{n}j}\Big)$ as the space for the limiting operator. Finally, in \cref{sec:limit}  we introduce the parameter $h$ and the partial isometry $T_h$. Then, after some preparatory results, we show the proof of \cref{teo:main}.

\section{Higher order differential structure on \texorpdfstring{$\Z^n$}{Zn}}\label{sec:CDS}
The discrete Hodge-Laplacian, as a higher dimensional analog of the usual graph Laplacian, was first introduced by Eckmann \cite{Ec45}. It was defined over combinatorial simplicial complexes and allows one to obtain a discrete Hodge Theory. Furthermore, Dodziuk showed in \cite{Do76} that such discrete Hodge Theory approximates the Hodge Theory of a Riemannian manifold if one considers finer triangularizations of the given manifold. In that setting, the spectral theory of Hodge-Laplacian or Hodge-Dirac operators has already been considered, see \cite{GW16,HJ13,Ch20,Ch08,ABDE21} for some particular questions that are not far away of the perspective of this manuscript. However, all those examples are concerned with such discrete differential operators on simplicial complexes and such structure becomes trivial on $\Z^n$ for cochains of order $2$ or higher. With that in mind, in this section we  introduce slightly more general assumptions, compared with that of a simplicial complex, that still allows us to have our discrete differential operators and, in the $\Z^n$ case, to define non-trivial higher order cochains.

\subsection{Abstract combinatorial differential structure}\label{subsec:abstract}

In order to motivate our definition below, we start by recalling the framework of discrete diferential operators on simplicial complexes.

An \emph{abstract simplicial complex} K  over a set $V$ is a collection of finite subsets of $V$ closed under inclusion. An element $F\in K$ is of dimension $i$ if it has cardinality $i+1$ and is called an $i$-simplex. An element in $F$ is called a vertice of $F$. By choosing an ordering of its vertices we say that we have an oriented simplex $[F]$. If there exist an even permutation transforming one ordering into another we say that they define the same orientation. If the permutation is odd, the orientations are considered to be oppposite and denoted by $\overline{[F]}$. 
The complex vector space $C^i(K)$ of $i$-cochains of $K$ is defined as the space of functions on $i$-simplices that satisfy
\begin{equation*}
f(\overline{[F]})=-f([F])\ .
\end{equation*}
Then the simplicial coboundary maps $d_i:C^i(K)\to C^{i+1}(K)$ are defined by
\begin{equation*}
d_if([x_0,x_1,\dots,x_{i+1}])=\sum_{j=0}^{i+1}(-1)^jf([x_0,\dots,\hat{x}_j,\dots x_{i+1}])
\end{equation*}
where $\hat{x}_j$ indicates that we have omitted the $j$-th vertex.

Motivated by this formalism but aiming to cover the case of $\Z^n$ we propose the following definition where we denote by $\mathbb{P}(V)$ the set of all subsets of a set $V$.
\begin{Definition}\label{CDC}
We call $X$ a \emph{combinatorial differential complex} of dimension $n$ over $V$ if there exist:
\begin{enumerate}
\item $\imath : X \to \mathbb{P}(V)$ such that for every $A\in\mathbb{P}(V)$ we have that either $\sharp \imath ^{-1}(A)=0$ or $\sharp \imath ^{-1}(A)=2$. This define an involution on $X$ by $\imath(s)=\imath(\overline{s})$. \label{noloops}
\item $X=\cup_{j=0}^n X^j$ with every $X^j$ non-empty, such that $v\in X^0 \Rightarrow \sharp \imath (v)=1$ and $e\in X^1 \Rightarrow \sharp \imath (e)=2$. Without loss of generality we can assume that $X^0=V\times \Z_2$ and hence define $\sgn:X^0\to\{-1,1\}$ that satisfies $\sgn (v)=-\sgn(\overline{v})$. \label{decomposition}
\item For $j\geq 1$, there exists $\partial : X^j\to \mathbb{P}(X^{j-1})$ that satisfies for $s\in X^j$, $\partial (\overline{s})=\overline{\partial(s)}$. Furthermore if $j\geq 2$ we ask $\partial (\partial(s))$ to be an involutive set while for $e\in X^1$, $\partial (e)=\{v_1,v_2\}$, to satisfy $\sgn{v_1}\sgn{v_2}=-1$.  \label{hypd20}
\end{enumerate}
\end{Definition}

If $r\in \partial(s)$ we say that $r$ is contained in $s$ according to orientation and denoted this by $r\subset s$. We assume that $r\subset s$ for only finitely many $s$. 

\begin{Example}
As intended, a simplicial complex is indeed a combinatorial differential complex. The $\imath$ operator takes an ordered simplex, forget its order and gives the subjacent set. The $\partial$ operator takes an ordered simplex and gives us the set of ordered faces. Note that we have taken the convention of giving two orientations to singletons. This should be understand as an inward copy of the vertex and an outward copy of the vertex. Then, an edge would be compose of two vertices, one taken with an outward signature, the origin, and one taken with the inward signature, the target. 
\end{Example}

\begin{Remark}
Although quite general, \cref{CDC} still does not admits some cases that one could be interested in studying. In particular, \cref{noloops} implies that we are working with non-oriented combinatorial differential complexes without loops. On one hand, the study of oriented discrete structures, that gives rise to non-symmetric differential operators, and hence non-self-adjoint Laplacian, is at the same time very interesting and not well explored from a purely spectral theoretical point of view. On the other hand, although for a infinite graph it seems rather harmless to not allow loops, in practice these kind of phenomena appears quite naturally when making a quotient and hence need to be dealt in order to develop a theory of \emph{periodic} combinatorial differential complexes. Since in this article we are confined to the $\Z^{n}$ case, we need not  to deal with such technicalities. Finally, one can also notice that the assumption that $r\subset s$ for only finitely many $s$ corresponds in the graph case, to be working with locally finite graphs. Again, this is not compulsory, but it suffices for this article aims. We refer to \cite{Ke15} for a survey of problems that could be of interest in the setting of combinatorial differential complexes with unbounded geometry.
\end{Remark}

Our task now is from \cref{CDC} to reproduce the usual discrete differential calculus. With that aim, we start by defining for $1\leq j \leq n$ the vector space of $j-$\emph{cochains}
\begin{equation*}
C^j(X)=\{f:X^j\to \C:f(\overline{s})=-f(s)\}
\end{equation*}
and the \emph{discrete exterior derivative} 
\begin{equation*}
d_j:C^{j}(X)\to C^{j+1}(X)
\end{equation*}
by
\begin{equation*}
d_jf(s)=\sum_{r \subset s}f(r) \ .
\end{equation*}

We stress that this definition for the exterior derivative relies only in what we consider to be the \emph{discrete differential structure}. However, to go one step further we need to introduce a measure on $X$:
\begin{equation*}
m:X\to \R_{>0}\quad ; \quad m(s)=m(\overline{s})\ .
\end{equation*} 
Let us denote by $C^j_c(X)$ the subspace of $j-$cochains of finite support. On it we can introduce an inner product by 
\begin{equation*}
\langle f,g\rangle_{X^j}:=\frac12\sum_{r\in X^j}m(r)f(r)\overline{g(r)}\ .
\end{equation*}
This inner product defines the Hilbert space $\ell^2(X^j)$ by taking the closure of $C^j_c(K)$ on $C^j(K)$. A sufficient condition for $d_j:\ell^2(X^j)\to \ell^2(X^{j+1})$ to be bounded is given by
\begin{equation*}
\sup_{s\in X^j}\{\sum_{s\subset r}\frac{m(r)}{m(s)} \}<\infty .
\end{equation*}
Hence, as a natural generalization of the degree of a vertex, we can define 
\begin{equation*}
\deg_m : X\to \R_{>0}\quad ;\quad \deg_m(s)=\sum_{s\subset r}\frac{m(r)}{m(s)}\ .
\end{equation*}
We turn now to determine the adjoint $d_j^*:\ell^2(X^{j+1})\to \ell^2(X^j)$ by computing for $f\in C_c^j(X)$ and $g\in C_c^{j+1}(X)$
\begin{align}
\langle d_jf,g\rangle_{X^{j+1}}=&\frac12 \sum_{s\in X^{j+1}}m(s)df(s)\overline{g(s)}\nonumber \\ 
=&\frac12 \sum_{s\in X^{j+1}}m(s)\overline{g(s)}\sum_{r\subset s}f(r) \label{paso1} \\
=&\frac12 \sum_{r\in X^{j}}m(r)f(r) \overline{\sum_{r\subset s}\frac{m(s)}{m(r)}g(s)}\ .\label{paso2}
\end{align}
Hence we have
\begin{equation}\label{dadjunto}
d_j^*g(r)=\sum_{r\subset s}\frac{m(s)}{m(r)}g(s)
\end{equation}
once we notice that for $r\in X^j$ such that $r\notin \partial(s)$ for every $s\in X^{j+1}$ the sum in \cref{dadjunto} vanishes and hence the step from \cref{paso1} to \cref{paso2} is fully justified.

The Hodge-Laplacian of order $j$ acting on $\ell^2(X^j)$ can then be defined by 
\begin{equation}\label{HodgeLaplace}
\Delta_j(X)=d_j^*d_j+d_{j-1}d_{j-1}^*\ .
\end{equation}

This structure defined for each order can be put together defining on $\ell^2(X)=\oplus_{j=0}^n\ell^2(X^j)$ the \emph{discrete exterior derivative}
\begin{equation*}
d:=\oplus_{j=0}^{n-1} d_j\ .
\end{equation*} Using this notation we can state the following lemma.

\begin{Lemma}\label{d20}
The operator
\begin{equation*}
d:\ell^2(X)\to\ell^2(X)
\end{equation*}
satisfies
\begin{equation*}
d^2=0\ .
\end{equation*}
\end{Lemma}

\begin{proof}
It is enough to note that for $f\in \ell^2(X^j)$ and $s\in X^{j+2}$ we have
\begin{equation*}
(d^2f)(s)=\[d_{j+1}(d_jf)\](s)=\sum_{r\subset s}d_jf(r)=\sum_{r\subset s}\sum_{t\subset r}f(t)=\sum_{t\in \partial(\partial (s))}f(t)=0
\end{equation*}
by recalling \cref{hypd20}. 
\end{proof}
The discrete Gauss-Bonnet operator is then given by $D:=d+d^*$. From \cref{d20} one can check that $D^2=\oplus_{j=0}^n \Delta_j(X):=\Delta(X)$, where $\Delta(X)$ is the full Hodge-Laplacian. 

To complete the necessary ingredients to state a Hodge theory on $X$ we only need the following Corollary.
\begin{Corollary}
\begin{equation*}
\Ker (\Delta_j(X))=\Ker (d_j)\cap \Ker (d_{j-1}^*)
\end{equation*}
\end{Corollary}
\begin{proof}
Let $f\in \Ker (\Delta_j)$. By \cref{HodgeLaplace} we have
\begin{equation}\label{kerhodge}
d_j^*d_jf=-d_{j-1}d_{j-1}^*f\ .
\end{equation}
Applying $d_j$ to both sides of \cref{kerhodge} we get from \cref{d20} that 
\begin{equation*}
d_jd_j^*d_jf=0\ .
\end{equation*}
From the properties of the adjoint we get that $f\in \Ker (d_j)$. Applying $d_{j-1}^*$ to both side of \cref{kerhodge} we get the other inclusion.
\end{proof}
Putting all this information together we get the desired \emph{Hodge decomposition} given by
\begin{equation*}
\ell2(X^j)=\Ran (d_{j-1})\oplus \Ker (\Delta_j(X)) \oplus \Ran (d_{j}^*)\ .
\end{equation*}
To close this \namecref{subsec:abstract}, let us ilustrate the supersymmetry, as per \cite{Th92}, exihibited by $d+d^*$. For this, we define 
\begin{equation*}
\ell^2(X)_{\text{even}}=\bigoplus_{j \text{ even}}\ell^2(X^j)\quad ; \quad \ell^2(X)_{\text{odd}}=\bigoplus_{j \text{ odd}}\ell^2(X^j)
\end{equation*}
which satisfy $\ell^2(X)=\ell^2(X)_{\text{even}}\oplus \ell^2(X)_{\text{odd}}$. Furthermore,
\begin{equation}\label{tauinvolution}
\tau f=\begin{cases}f &\text{ if }f\in\ell^2(X)_{\text{even}}\\
-f &\text{ if }f\in\ell^2(X)_{\text{odd}}\\
\end{cases}
\end{equation}
defines an involution on $\ell^2(X)$ such that the restriction of $D$ to $\ell^2(X)_{\text{odd}}$ anticommutes with $\tau$. Given $m\geq 0$, the Dirac operator on $X$ is given by
\begin{equation*}
\mathbb{D}_m=D+\tau m\ .
\end{equation*}
For simplicity, in this paper  we consider only the massless case $m=0$.

\subsection{Application to \texorpdfstring{$\Z^n$}{Zn}}\label{subsec:defzn}

As indicated in \cref{sec:intro}, our aim is to endow $\Z^n$ with a discrete differential structure that, in contrast with the simplicial one, is not trivial for dimensions $j\geq2$. 
\begin{Theorem}
There exist a \emph{combinatorial differential structure} over $\Z^n$, as per \cref{CDC}, such that $\ell^2(X^j)$ is infinite dimensional for every $1\leq j \leq n$.
\end{Theorem}

Let us denote by $\{\delta_1,\dots,\delta_n\}$ the canonical basis of $\Z^n$. An hyper-cube $s\subset \Z^n$ of dimension $j$ has $2^j$ vertices. In particular, there exists a unique way of describing $s$ by a $2^j$-tuple
\begin{equation*}
(x_1,x_2,\dots x_{2^j-1},x_{2^j})
\end{equation*}
that satisfies $x_{i+1}-x_i=\pm\delta_{l}$ with $x_{i+1}-x_i=\delta_{l_i}$ and $l_i<l_{i+1}$ for $1\leq i\leq j$. This is considered an orientation on $s$ and by convention we refer to this orientation as positive. The opposite orientation is given by 
\begin{equation*}
(x_{j+1},x_j,x_{j-1},\dots,x_2,x_1,x_{2^j},x_{2^j-1},\dots,x_{j+2})=(\tilde{x}_1,\dots,\tilde{x}_{2^j})
\end{equation*} and satisfies $\tilde{x}_{i+1}-\tilde{x}_i=-\delta_{\tilde{\ell}_i}$ and $\tilde{\ell}_i>\tilde{\ell}_{i+1}$ for $1\leq i\leq j$. Hence, an oriented hyper-cube $s$ of dimension $j$ in $\Z^n$ can be  described by a ``base'' point in $\Z^n$ which given by the  first element in the $2^j$-tuple, and   $j$ elements of the canonical basis ordered on either increasing or decreasing fashion.   Increasing and decreasing correspond to  the positive and negative orientations, respectively.
Denoting the base point by $\lfloor s\rfloor$ we thus write 
\begin{equation}\label{sequiv}
s\equiv (\lfloor s\rfloor;\delta_{l_1},\dots,\delta_{l_j})\ . 
\end{equation}
In particular, in the first hyperoctant a positive oriented hyber-cube $s$  satisfies that  $\lfloor s\rfloor$ is the smallest vertice. In contrast, for a negatively oriented one $\lfloor s\rfloor$ is its biggest vertice.

For $j=1$, the notation $(x;\delta_\ell)$ conflates $(x,x+\delta_{\ell})$ with $(x,x-\delta_{\ell})$. In the first case the singleton $\delta_{\ell}$ needs to be understood as ordered in an increasing fashion while in the later case in decreasing order. For $j=0$, we repeat our convention of considering each $x \in \Z^n$ with an outward and inward orientation. We denote these orientation by $(x;-)$ and $(x;+)$ respectively.

Let $P^{j,n} $ be the set of functions 
\begin{equation}\label{Pjndef}
P^{j,n}:=\{I:\{1,\dots,j\}\to\{1,\dots,n\}; I \,\,{\rm strictly \,\, monotone}\}
\end{equation}
 By \cref{sequiv} a simplex $s$ of dimension $j$ defines an element $\hat s$ in $P^{j,n} $. Hence we can also use the following notation 
\begin{equation}\label{shat}
(\lfloor s\rfloor;\delta_{\hat{s}(1)},\dots,\delta_{\hat{s}(j)})\equiv(\lfloor s\rfloor;\hat{s})\ .
\end{equation}
 On $P^{j,n}$ we can define an involution by
\begin{equation*}
\hat{s}^*(i)=\hat{s}(j-i+1)\ ,
\end{equation*}
and a signature by taking $+1$ if it is increasing and $-1$ if it is decreasing. By assumption we have that $\sgn(s)=\sgn(\hat{s})$. We set
\begin{equation*}
\lceil s \rceil:=\lfloor s \rfloor +\sgn(s)\sum_{i=1}^j\delta_{\hat{s}(i)}\ .
\end{equation*}
This notation allow us to relate an hyper-cube  and its opposite orientation by
\begin{equation*}
\overline{s}=(\lceil s \rceil;\hat{s}^*)=:(-1)s\ .
\end{equation*}

We set $X(\Z^n)$ as the set of oriented hyper-cubes. It is easy to see that it satisfies \cref{noloops} of \cref{CDC} by taking $\imath:X(\Z^n)\to \mathbb{P}(\Z^n)$ which for every hyper-cube gives its set of vertices.
Denoting by $X^j$ the set of oriented hyper-cubes of dimension $j$ we can write
\begin{equation*}
X(\Z^n)=\cup_{j=0}^n X^j
\end{equation*}
in order to satisfy \cref{decomposition} of \cref{CDC}. We have left to define 
\begin{equation*}
\partial : X^{j+1}\to \mathbb{P}(X^j)\ .
\end{equation*}
With that aim, starting by the application \cref{shat} and for $1\leq i_0 \leq j$, we define $\prescript{}{i_0}{\hat{s}}\in P^{j-1,n}$ by
\begin{equation*}
\prescript{}{i_0}{\hat{s}}(i)=\begin{cases}
\hat{s}(i) &\text{ if }i<i_0\ ;\\
\hat{s}(i+1) &\text{ if }i_0\leq i\ .
\end{cases}
\end{equation*}
For $j\geq 2$ and $s\in X^j\subset X(\Z^n)$ we define
\begin{equation}\label{sborde}
\partial (s) := \cup_{i=1}^j \{(-1)^{j-i} (\lfloor s \rfloor;\prescript{}{i}{\hat{s}}) \}\bigcup\cup_{i=1}^j\{(-1)^{i} (\lceil s \rceil;(\prescript{}{i}{\hat{s}})^*)  \}\ .
\end{equation}

Let us illustrate this definition by an example. Let $s$ be a cube, $s=(x;\delta_1,\delta_2,\delta_3)$. Then, $\partial (s)$ is composed of $6$ faces. Three of these faces contain $x=\lfloor s \rfloor$ :
\begin{equation}
(x;\delta_1,\delta_2) \, ;\, \overline{(x;\delta_1,\delta_3)}\, ;\, (x;\delta_2,\delta_3)\ ;
\end{equation}
while three contain $x+\delta_1+\delta_2+\delta_3=\lceil s \rceil$ :
\begin{equation}
\, \overline{(x+\delta_1+\delta_2+\delta_3;\delta_3,\delta_2)} \, ;\, (x+\delta_1+\delta_2+\delta_3;\delta_3,\delta_1)\, ;\overline{(x+\delta_1+\delta_2+\delta_3;\delta_2,\delta_1)}\ .
\end{equation}

We complete our definition for $s\in X^1$ by $\partial (x,y)=\{(x;-),(y;+)\}$. We can now state the result that would show that this definition is in accordance with \cref{hypd20} of \cref{CDC}

\begin{Proposition}\label{todofunciona}
Let $s\in X^j\subset X(\Z^n)$, with $j\geq2$. Then $\partial (\overline{s})=\overline{\partial(s)}$ and if $r\in\partial(\partial(s))$, we have that $\overline{r}\in\partial(\partial(s))$.
\end{Proposition}

The proof of \cref{todofunciona} is rather long and elementary so we posponed to the \cref{apendice}. 

We have hence endowed $\Z^n$ with a \emph{combinatorial differential structure}. By extension, we also have this structure for $h\Z^n:=\{(hz_1,\hdots,hz_n):z_l\in\Z\}$ with  $h>0$. We denote

\begin{equation*}
X(h\Z^n)=\bigcup_{j=0}^n X_h^j\ .
\end{equation*}

The only ingredient missing for having the Gauss-Bonnet operator is the measure $m:X(h\Z^n)\to \R_{>0}$. For $s\in X_h^j$ we set 

\begin{equation*}
m(s)=h^{-2j}\ .
\end{equation*}
This measure is chosen in order to ensure that the norm of $\delta_{s_0}$, the delta function over $s_0\in X_h^j$,  correspond to the volume of the (shrunken) $j$-hyper-cube. Indeed we can notice that
\begin{equation*}
||\delta_{s_0}||=\left(\sum_{s\in X}m(s)|\delta_{s_0}(s)|^2\right)^{\frac12}=\left(m(s_0)\right)^{\frac12}=h^{-j}\ .
\end{equation*}
It follows that we have the following combinatorial differential operators on $\ell^2(X(h\Z^n))$:
\begin{equation}\label{d_h}
(d_hf)(s)=\sum_{r \subset s}f(r)\quad ; \quad( d_h^*f)(s)=\frac1{h^2} \sum_{s \subset r}f(r)\ .
\end{equation}
One can check, using for example \cref{lemmaseccion3,le00}, that the spectrum of $d_h+d_h^*$ is given by
\begin{equation*}
\left[-\frac{\sqrt{4n}}h,\frac{\sqrt{4n}}h\right]\ .
\end{equation*}

For ease of notation, when $h=1$ we denote the corresponding Hilbert space by $\ell^2(X(\Z^n))$ and drop the $h$ subscript.

\subsection{The particular case \texorpdfstring{$\Z^1$}{Z1}}

The study of 1-dimensional discrete Dirac operator clearly outweigh its higher dimensional counterpart as can be attested for example by \cite{BMT20,CIKS20,POO21} and reference therein. In the topic of the continuum limit, \cite{SU21x} focus in the $n=2$ and argue that for $n=1$ the discrete Dirac operator on $\ell^2(\Z)^2$ is given by
\begin{equation*}
\mathbb{D}_m=\begin{pmatrix}
m & D_{-}\\
D_+ &-m
\end{pmatrix}
\end{equation*}
where $D_{+}f(x)=f(x+1)-f(x)$ and $D_{-}=D_{+}^*$, is unitarily equivalent to a Schrödinger operator on $\ell^2(\Z)$. From this, one can use the result from \cite{NT21} to study its continuum limit. From another point view, in contrast with the modifications needed for bigger dimension, \cite{CGJ22x} shows that for $n=1$ the continuum limit holds in the norm resolvent sense for $\mathbb{D}_m$. These results coincide with the present work in the sense that $\mathbb{D}_m$ is unitarily equivalent to $d+d^*+m\tau$, where $\tau$ was defined on \cref{tauinvolution}. This equivalence can be implemented by the unitary operator $\mathbb{U}:\ell^2(X(\Z))\to \ell^2(\Z)^2$ is given by
\begin{equation*}
(\mathbb{U}f)(x)=(f(\{x,+\}),f(x,x+1))\ .
\end{equation*}

\section{The operator \texorpdfstring{$d$}{d} acting on discrete differential forms}\label{sec:hodgedirac}

In this section we construct  a unitary representation of  $\ell^2(X(\Z^n))$ which will be useful for the subsequent computations. For this construction we use standard notation from multilinear algebra and differential geometry.

First, for $0\leq j\leq n$ let $\bigwedge^j(\Z^n)$ be the vector space of alternating $j$-linear complex valued functions on $(\Z^n)^j$, where we identify $\bigwedge^0(\Z^n)$ with $\C$.  We refer to a $\omega\in\bigwedge^j(\Z^n) $ as a $j$-form. We define $dx^l\in\bigwedge^1(\Z^n)$ by
\begin{equation*}
dx^l(\delta_k)=\begin{cases}
1 &\text{ if }l=k\\
0 &\text{ if }l\neq k
\end{cases}
\ .
\end{equation*} 

A basis for $\bigwedge^1(\Z^n)$ is given by $\{dx^1,\hdots,dx^n\}$. The wedge product $\wedge:\bigwedge^k(\Z^n)\times\bigwedge^j(\Z^n)\to\bigwedge^{k+j}(\Z^n)$ for $\eta\in\bigwedge^k(\Z^n)$ and $\omega\in\bigwedge^j(\Z^n)$ is given by
\begin{equation*}
 \eta\wedge\omega(\mu_1,\hdots,\mu_{k+j})=\frac1{k!j!}\sum_{\sigma \in S_{k+j}}{\rm sgn}(\sigma)\eta(\mu_{\sigma(1)},\hdots,\mu_{\sigma(k)})\omega(\mu_{\sigma(k+1)},\hdots,\mu_{\sigma(k+j)})\ .
\end{equation*}
Then, a basis for $\bigwedge^j(\Z^n)$ is given by
\begin{equation*}
\{\omega\in \bigwedge\nolimits^j(\Z^n) : \omega=dx^{l_1}\wedge\hdots\wedge dx^{l_j} \text{ and }l_1<\dots<l_j \}\ .
\end{equation*}

Let $P^{j,n}_+$ be the elements in $P^{j,n}$  (defined in \cref{Pjndef}) that are strictly increasing. For $I\in P^{j,n}_+$ we define $dx^I:=dx^{I(1)}\wedge\hdots\wedge dx^{I(j)}$. 
Using this basis we can define an inner product on $\bigwedge^j(\Z^n)$ by:
\begin{equation}\label{26sep22}
\langle dx^I;dx^{I'}\rangle_{\bigwedge\nolimits^j(\Z^n)}:=\begin{cases}
1 &\text{ if } I=I'\\
0&\text{ if else }
\end{cases}\ .
\end{equation}

Let $\Omega^j(\Z^n):=\{\omega:\Z^n\to\bigwedge^j(\Z^n)\}$ be the vector space sections over $\bigwedge^j(\Z^n)$. We denote by $\Omega_c^j(\Z^n)$ the subspace of compactly supported sections.
Each $\omega$ in $\Omega^j(\Z^n)$ can be written by
\begin{equation}\label{14dec22}
\omega(\mu)=\sum_{I\in P^{j,n}_+} \omega_{I}(\mu)  dx^{I},
\end{equation}
where $\omega_{I}:\Z^n\to\C$. On $\Omega_c^j(\Z^n)$ we consider the inner product induced by the inner product on $\bigwedge^j(\Z^n)$ given by \cref{26sep22}. Using \cref{14dec22} one can check that
\begin{equation}\label{omegajinterno}
\langle \omega , \eta\rangle_{\Omega_c^j(\Z^n)}=\sum_{\mu\in\Z^n}\sum_{I\in P^{j,n}_+} \omega_{I}(\mu)\overline{\eta_{I}(\mu)}\ .
\end{equation}
The completion of $\Omega_c^j(\Z^n)$ under the norm defined by \cref{omegajinterno} is $\ell^2(\Z^n;\bigwedge^j(\Z^n))$. 

Let us introduce the operator  $\tilde d$  in $\Omega^j(\Z^n)$ using  the standard ideas of differential forms. First, for a function $\omega:\Z^n\to\C$  set 
$$
\mathcal{D}_{l}\omega(\mu)= \omega(\mu+\delta_l)-\omega(\mu).
$$
Let $\tilde d_{0}:\Omega^0(\Z^n)\to \Omega^1(\Z^n)$ be the exterior derivative defined by
\begin{equation}\label{15sep22}   
\tilde d_{0} \omega=\sum_{l=1}^n(\mathcal{D}_{l} \omega) d x_l.
\end{equation}

We are interested in comparing $d$ from \cref{d_h} to $\tilde d$. With that in mind, let us define, for every $0\leq j\leq n$, the unitary operator ${U}_j:\ell^2(X^j)\to\ell^2(\Z^n;\bigwedge^j(\Z^n))$ by
\begin{equation}\label{14dec22b}
({U}_jf)(\mu):=\sum_{I\in P^{j,i}_+} f(\mu;\delta_{I(1)},\hdots, \delta_{I(j)}) \, dx^I\ .
\end{equation}
Noticing that for $f\in \ell^2(X^0)$
\begin{align}
[(U_1\circ d_0)f](\mu)=\sum_{i=1}^j (d_0f)(\mu;\delta_j)dx^{i}=\sum_{i=1}^j (d_0f)(\mu,\mu+\delta_i)dx^{i}=&\sum_{i=1}^j (f(\mu+\delta_i)-f(\mu))dx^{i}\nonumber \\
=&\sum_{i=1}^j \mathcal{D}_if(\mu)dx^{i}\label{Udconmuta}
\end{align}
we obtain that $U_1\circ d_0=\tilde{d}_0\circ U_0$.

The operator
 $\tilde d_j:\Omega^j(\Z^n)\to \Omega^{j+1}(\Z^{n})$ is defined by
\begin{equation}\label{defd}
\tilde d_j (\sum_{I\in P^{j,n}_+}\omega_I dx^I)=\sum_{I\in P^{j,n}_+}(\tilde d_0 \omega_I)\wedge dx^I.
\end{equation}
Arguing as in \cref{Udconmuta} one can check that we have the following commutative diagram 
\begin{figure}[h]
\begin{center}
\begin{tikzcd}
\ell^2(X^j) \arrow[r, "{{U}_j}"] \arrow[d, "d_j"]
& \ell^2(\mathbb{Z}^n;\bigwedge^j (\mathbb{Z}^n)) \arrow[d, "{\tilde{d}_j}"] \\
\ell^2(X^{j+1}) \arrow[r, "{{U}_{j+1}}"]
& \ell^2(\mathbb{Z}^n;\bigwedge^{j+1} (\mathbb{Z}^n)).
\end{tikzcd}
\end{center}
\end{figure}

To have a global perspective we define the \emph{total exterior derivative} $\tilde{d}$ as the off--diagonal operator given by
\begin{equation*}
\tilde{d}:\oplus_{j=0}^{n-1}\Omega^j(\Z^n)\to\oplus_{j=1}^{n}\Omega^j(\Z^n)\ .
\end{equation*}
and from the family of unitary operators $U_j$, the diagonal operator
\begin{equation*}
U:\ell^2(X(\Z^n))\to\bigoplus_{j=0}^n \ell^2(\Z^n;\bigwedge\nolimits^j(\Z^n))\ .
\end{equation*}
Then, we can summarize the results so far by 
\begin{equation}\label{Udconmutagrande}
\tilde{d}\circ U= U\circ d\ .
\end{equation}

We turn now our attention to $\tilde{d}^*$. For this, let us consider $I\in P^{j,n}_+$ and define $J_I:\{j+1,\dots,n\}\to\{1,\dots,n\}\setminus\{I(1),\dots,I(j)\}$ to be the only strictly increasing function between those two sets. Then we denote by $IJ_I$ 
the permutation given by
\begin{equation*}
IJ_I(i)=\begin{cases}
I(i) \quad &i\leq j\\
J_I(i) \quad &j< i
\end{cases}\ .
\end{equation*}
 The \emph{Hodge star operator} is defined as the operator $*:\bigwedge^j(\Z^n)\to\bigwedge^{n-j}(\Z^n)$ determined by 
$$
*dx^I={\rm sign}(IJ_I) dx^{J_I}\ .
$$
Then, we can compute the adjoint operator $\tilde d^*:\Omega^{j+1}(\Z^n)\to \Omega^{j}(\Z^{n})$ as
\begin{equation}\label{22sep22}
\tilde d^*=(-1)^{ni+1}*\tilde d*.
\end{equation}
Let us denote by $\Delta$ the discrete Laplacian on $\Z^n$ given by:
$$
(\Delta f)(\mu)=\sum_{l=1}^n f(\mu+\delta_l)+f(\mu-\delta_l)-2f(\mu)\ , 
$$
and consider the bounded operator $\tilde\Delta$ on $\ell^2(\Z^n;\bigwedge^j(\Z^n))$ defined by 
\begin{equation*}
\tilde\Delta\left(\sum\omega_{I}dx^I\right)\nonumber
=\sum (\Delta \omega_{I})dx^I\ .
\end{equation*}

The following result shows that the Hodge-Laplacian on $\Z^n$ is unitary equivalent to a family of usual Laplacians on $\Z^n$.

\begin{Lemma}\label{lemmaseccion3}
\begin{equation*}
(\tilde d+\tilde d^*)^2=\tilde \Delta
\end{equation*}
\end{Lemma}

\begin{proof} 

First notice that by \cref{Udconmutagrande} we know that $\tilde{d}^2=0$. It follows that $(\tilde d+\tilde d^*)^2=\tilde d\tilde d^*+\tilde d^*\tilde d$ and hence we will prove $(\tilde d\tilde d^*+\tilde d^*\tilde d)=\tilde \Delta$.  Consider the quadratic form of the operator $\tilde\Delta$ given by
\begin{align}\label{23sep22}
\langle \sum_I(\Delta \omega_I) dx^I;\sum_I\omega_Idx^I  \rangle=&\sum_{\mu\in \Z^n}\sum_I (\Delta \omega_I(\mu)) 
\overline{\omega_I(\mu) }\nonumber\\
=&\sum_{\mu\in \Z^n}\sum_I \sum_{\alpha=1}^n((\mathcal{D}_{\alpha}^2\omega_I)(\mu-e_\alpha))\overline{\omega_I(\mu)}\nonumber\\
=&\sum_{\mu\in \Z^n}\sum_I \sum_{\alpha=1}^n|  \mathcal{D}_{\alpha}\omega_I (\mu)|^2\ .
\end{align}

Now, let us compute the quadratic form of $\tilde d\tilde d^*+\tilde d^*\tilde d$. Using \cref{22sep22} we obtain that 
\begin{align}\label{23sep22b}
\langle (\tilde d\tilde d^*)\omega, \omega \rangle &= \langle \tilde d*\omega;\tilde d*\omega\rangle\nonumber\\
&=\langle \tilde d\sum_I{\rm sign}(IJ_I) \omega_Idx^J;\tilde d\sum_I{\rm sign}(IJ_I) \omega_Idx^{J_I}\rangle\nonumber\\
&=\langle \sum_I{\rm sign}(IJ_I) \sum_{\alpha=1}^n{  \mathcal D_{\alpha}\omega_I}dx^\alpha\wedge dx^{J_I};\sum_I{\rm sign}(IJ_I) \sum_{\alpha=1}^n{  \mathcal D_{\alpha}\omega_I}dx^\alpha\wedge dx^{J_I}\rangle\nonumber\\
&=\sum_{\mu\in \Z^n}\sum_I\sum_{\substack{\alpha\neq J_I(i)\\ j+1\leq i\leq n}}|   \mathcal D_{\alpha}\omega_I (\mu)|^2\nonumber\\
&=\sum_{\mu\in \Z^n}\sum_I\sum_{i=1}^j|  \mathcal{D}_{I(i)}\omega_I (\mu)|^2.
\end{align}

Similarly 

\begin{align}\label{23sep22c}
\langle (\tilde d^*\tilde d)\omega, \omega \rangle &= \langle \tilde d\omega;\tilde d\omega\rangle\nonumber\\
&=\langle \sum_I\sum_{\alpha=1}^n{  \mathcal D_{\alpha}\omega_I}dx^\alpha\wedge dx^J;\sum_I\sum_{\alpha=1}^n{  \mathcal D_{\alpha}\omega_I}dx^\alpha\wedge dx^J\rangle\nonumber\\
&=\sum_{\mu\in \Z^n}\sum_I\sum_{\substack{\alpha\neq I(i)\\ 1\leq i\leq j}}|    \mathcal D_{\alpha}\omega_I (\mu)|^2\nonumber\\
&=\sum_{\mu\in \Z^n}\sum_I\sum_{i=j+1}^n|  \mathcal{D}_{J_I(i)}\omega_I (\mu)|^2.
\end{align}

Adding \cref{23sep22b,23sep22c} we obtain \cref{23sep22}.
\end{proof}

This result is useful for the present article because it allows the following representation of the resolvent of $d+d^*$:
\begin{Corollary}\label{le00} For ${\rm Im}(z)\neq0$ the resolvent operators of $\tilde d+\tilde d^*$ and $\tilde\Delta$ satisfies 
$$(\tilde d+\tilde d^*-z)^{-1}=(\tilde d+\tilde d^*+z)(\tilde\Delta-z^2)^{-1}.$$
\end{Corollary}

\section{Continuum limit}\label{sec:limit}

\subsection{Introducing the mesh size \texorpdfstring{$h$}{h}}
Now we are ready to introduce the parameter $h$. As in the last part of section \cref{sec:CDS} we consider the spaces $\bigoplus_{j=0}^n \ell^2( h\Z^n;\bigwedge^j(h \Z^n))$  and $ \bigoplus_{j=0}^n \ell^2\Big(h\Z^n;\C^{\binom{n}{j}}\Big)$ with internal products given respectively  by 
\begin{align*}
\langle \omega, \eta\rangle_{\ell^2( h\Z^n;\bigwedge^{j}(h \Z^n))}&=\frac1{h^{2j}}\sum_{\mu\in h\Z^n;I\in P^{j,n}_+}\omega_I(\mu)\overline{\eta}_I(\mu), \\
\langle f, g\rangle_{\ell^2\Big(h\Z^n;\C^{\binom{n}{j}}\Big)}&=\sum_{\mu\in h\Z^n;1\leq l\leq \binom{n}{j} }f_l(\mu)\overline{g}_l(\mu).
\end{align*}
To define a unitary operator between these spaces we introduce  the lexicographic order in $P_+^{j,n}$, which is given by: $I<I'$ if and only if there exists $1\leq l \leq j$ such that 
\begin{align*}
I(i)=&I'(i),\quad i<l,\\
I(l)<&I'(l).
\end{align*}
Therefore  we have  $P_+^{j,n}=\{I_1^j<\dots<I_{\binom{n}{j}}^j\}$. For  $0\leq j\leq n$, using \cref{14dec22}, we set 
\begin{equation*}
\tilde U_{j,h}:\ell^2(h\Z^n;\bigwedge\nolimits^j(h\Z^n))\to  \ell^2\Big(h\Z^n;\C^{\binom{n}{j}}\Big) \ ,
\end{equation*} 
by
\begin{equation}\label{14dec22c}
(\tilde U_{j,h} \omega)_{l}(\mu):=\frac1{h^{{j}}}\omega_{I^j_l}(\mu),\qquad 1\leq l \leq \binom{n}{j}.
\end{equation} 
This gives a unitary operator whose adjoint is given by
\begin{equation}\label{tilde}
(\tilde{U}_{j,h}^* f)(\mu):=h^{j}\sum_{l=1}^{\binom{n}{j}} f_l(\mu)dx^{I^j_l}
\end{equation}

Using the same definitions of \cref{d_h,15sep22,defd} together with \cref{14dec22b,14dec22c} give us the comutative diagram
\begin{figure}[h]
\centering
\begin{tikzcd}
\ell^2(hX^j) \arrow[r, "{{U}_j}"] \arrow[d, "d_{j,h}"]
& \ell^2(h\mathbb{Z}^n;\bigwedge^j (h\mathbb{Z}^n)) \arrow[r, "\tilde {U}_{j,h}"] \arrow[d, "{\tilde{d}_{j,h}}"]
& \ell^2(h\mathbb{Z}^n;\mathbb{C}^{\binom{n}{j}}) \arrow[d, "{\tilde U_{j+1,h} \tilde{d}_{j,h}  \tilde {U}_{j,h}^{*}}"] \\
\ell^2(hX^{j+1}) \arrow[r, "{{U}_{j+1}}"]
& \ell^2(h\mathbb{Z}^n;\bigwedge^{j+1} (h\mathbb{Z}^n)) \arrow[r, "\tilde U_{j+1,h}"]
&\ell^2(h\mathbb{Z}^n;\mathbb{C}^{\binom{n}{j+1}})
\end{tikzcd}
\end{figure}

Defining the diagonal operator $\tilde U_h:\bigoplus_{j=0}^n \tilde U_{j,h}$ 
we are led to study the convergence of the operator $$H_h:=\tilde U_h U(d_h+d_h^*)U^{*} \tilde U_h^{*}: \bigoplus_{j=0}^n \ell^2\Big(h\Z^n;\C^{\binom{n}{j}}\Big)\to \bigoplus_{j=0}^n \ell^2\Big(h\Z^n;\C^{\binom{n}{j}}\Big)\ .$$ This operator can be written in the operator-valued matrix form
\begin{equation*}
H_h=\tilde {U}_h U \begin{pmatrix}
0 & d_{0,h}^*& 0&\hdots&0&0&0\\
d_{0,h} & 0 & d_{1,h}^*&\hdots&0&0&0\\
0 &  d_{1,h} & 0&\hdots&0&0&0\\
\vdots&\vdots&\vdots&\ddots&\vdots&\vdots&\vdots\\
0&0&0&\cdots&0&d_{j-1,h}^*&0\\
0&0&0&\cdots&d_{j-1,h}&0&d_{j,h}^*\\
0&0&0&\cdots&0&d_{j,h}&0
\end{pmatrix}U^{*} \tilde{U}_h^{*}.
\end{equation*}

Now we turn to construct the passage to the continuum. For this, we denote an element in $\bigoplus_{j=0}^n \ell^2\Big(h\Z^n;\C^{\binom{n}{j}}\Big)$ by $f_{j,l}$ for  $1\leq j\leq n; 1\leq l\leq \binom{n}{j}$ and set 
 $
F:\bigoplus_{j=0}^n \ell^2\Big(h\Z^n;\C^{\binom{n}{j}}\Big)\to \bigoplus_{j=0}^nL^2\Big(h^{-1}\T^n;\C^{\binom{n}{j}}\Big)$ 
 \begin{equation*} 
 (F f)_{j,l}(\xi):=h^n\sum_{\mu\in h\Z^n}e^{-2\pi i \xi\cdot\mu}f_{j,l}(\mu)\ , \qquad 0\leq j\leq n,\quad 1\leq l\leq \binom{n}j.
 \end{equation*} 

The next lemma describe $F H_h F^*$ as matrix-value multiplication operator. This characterization will enable us to adapt use the results from \cite{NT21} to obtain the desired limit.

\begin{Lemma}\label{coeficientesHh}The operator $F H_h F^*$ can be  written as a $2^n\times2^n$- matrix   whose coefficients are linear combinations of elements of the form: 
\begin{align*}
a_{h,l}(\xi):=\frac{(-1+e^{-2\pi ih\xi_l})}h,\quad 1\leq l \leq n.
\end{align*}
\end{Lemma}
\begin{proof}
 For  fixed $0\leq j\leq n$ let us consider the operator  $\tilde d_j\tilde U_{j,h}^*$ applied to $f=(f_{j,l})\in \ell^2\Big(h\Z^n;\C^{\binom{n}{j}}\Big)$.  Using \cref{defd,tilde}
\begin{align*}
(\tilde d_{j,h}\tilde U_{j,h}^* f)(\mu)&=h^{j}\sum_l^{\binom{n}{j}}   (\tilde d_{0,h} f_{j,l})(\mu)\wedge dx^{I^j_l}\\
&=h^{j}\sum_l^{\binom{n}{j}}  \sum_{\alpha=1}^n  {(f_{j,l}(\mu+h\delta_\alpha)-f_{j,l}(\mu))}dx^\alpha\wedge dx^{I^j_l}\\
&=h^{j}\sum_{1\leq \tilde  l\leq {\binom{n}{j}}}  \left(\sum_{\substack{
\alpha,l \\  dx^\alpha\wedge dx^{I^j_l}=(\pm)dx^{ I^{j+1}_{\tilde l}}
}}  {(f_{j,l}(\mu+h\delta_\alpha)-f_{j,l}(\mu))}\right) (\pm)dx^{ I^{j+1}_{\tilde l}}.
\end{align*}
 From this computation and \cref{14dec22c} we see that the $(\tilde l,l)$ coefficient of the $\binom{n}{j+1}\times \binom{n}{j}$ operator matrix  $\tilde U_{j+1,h}\tilde d_j\tilde U_{j,h}^*$ is 
 $$
\sum_{\substack{
1\leq \alpha\leq n \\  dx^\alpha\wedge dx^{I^j_l}=(\pm)dx^{ I^{j+1}_{\tilde l}}
}}  \frac{(f_{j,l}(\mu+h\delta_\alpha)-f_{j,l}(\mu))}h.
 $$
 We finish the proof  by recalling that $F(f_{j,l}(\cdot+h\delta_\alpha)-f_{j,l}(\cdot))=h\,a_{h,\alpha}Ff_{j,l}$.
\end{proof}

Our problem is  to study the convergence of $(F H_h F^*-z)^{-1}$. Using \cref{le00} we immediately  see   that for  ${\rm Im}(z)>0$
 \begin{align}\label{eq:resolventefibrado1}
(F H_h F^*-z)^{-1}(\xi)
&=\frac{1}{r_z(\xi)}
FH_hF^*+\frac{z}{r_z(\xi)}
\end{align}
 where
 \begin{equation*}
r_z(\xi):=\sum_{l=1}^n|a_{h,l}(\xi)|^2-z^2.
\end{equation*}
 
\subsection{The limit operator and some auxiliary results}\label{limitoperator}
For the limit operator we will consider the operator exterior derivative $\d$ in $\Omega(\R^n)$. Then, we define the operator $H$ in $\bigoplus_{j=0}^nL^2\Big(\R^n;\C^{\binom{n}{j}}\Big)$ with domain
 $\bigoplus_{j=0}^n\mathcal H^1\Big(\R^n;\C^{\binom{n}{j}}\Big)$, the first order Sobolev space, and is given by the operator-valued matrix 
 \begin{equation*}
H=\begin{pmatrix}
0 & \d_0^*& 0&\hdots&0&0&0\\
\d_0 & 0 & \d_1^*&\hdots&0&0&0\\
0 &  \d_1 & 0&\hdots&0&0&0\\
\vdots&\vdots&\vdots&\ddots&\vdots&\vdots&\vdots\\
0&0&0&\cdots&0&\d_{j-1}^*&0\\
0&0&0&\cdots&\d_{j-1}&0&\d_j^*\\
0&0&0&\cdots&0&\d_j&0
\end{pmatrix}
\end{equation*}

Defining the Fourier transform $\F:\bigoplus_{j=0}^nL^2\Big(\R^n;\C^{\binom{n}{j}}\Big)\to \bigoplus_{j=0}^nL^2\Big(\R^n;\C^{\binom{n}{j}}\Big)$ by 
$$
(\F f)_{j,l}(\xi)=\int_{\R^n}e^{-2\pi i x\cdot\xi}f_{j,l}(x)dx, \quad 1\leq j\leq n; 1\leq l \leq \binom{n}j, 
$$
 in a completely analogous manner as we did before, we can  see  that for ${\rm Im}(z)>0$
 \begin{equation}\label{eq:resolventefibrado2}
(\mathcal{F} H \mathcal{F}^*-z)^{-1}(\xi)=\frac{1}{R_z(\xi)}\mathcal{F} H \mathcal{F}^*+\frac{z}{R_z(\xi)}
\end{equation}
 where $
R_z(\xi):=4\pi^2|\xi|^2-z^2$. This time the coefficients of the matrix of $\mathcal{F} H \mathcal{F}^*$ are obtained by replacing $a_{h,l}$ of \cref{coeficientesHh} by
\begin{equation*}
A_l(\xi):=2i\pi \xi_l\quad 1\leq l \leq n \ .
\end{equation*}

Next, using the construction from \cite{NT21} define the partial isometry $P_h:\bigoplus_{j=1}^n \Big(L^2(\R^n);\C^{\binom{n}{j}}\Big)\to \bigoplus_{j=1}^n \Big(\ell^2(h\Z^n);\C^{\binom{n}{j}}\Big)$ 
$$
(P_hu)(\mu)=h^{-2}\int_{\R^n} \overline{\varphi_{h,\mu}(x)}u(x)dx,
$$
where $\varphi_{h,\mu}(x):=\varphi(h^{-1}(x-\mu))$, and $\varphi$ is any function in $\mathcal{S}(\R^n)$ satisfying
\begin{align}\label{condi_phi}
&\sum_{\mu\in\Z^n}|\hat{\varphi}(\xi+\mu)|^2=1,\quad \xi\in\R^n,\\\label{condi_phi2}
&{\rm supp}(\hat{\varphi})\subset (-1,1)^n.
\end{align}

Set also $Q_h:=F P_h\F^*$.  

The following two lemmas are small modifications of \cite[Lemma 2.2, Lemma 2.3]{NT21}. For ease of reading, we include the main ideas of the proofs here.
\begin{Lemma}\label{le0} For each $z\in\C\setminus \R$, there exists a positive constant $C$ such that 
$$
\|(1-P_h^*P_h)(H-z)^{-1}\|\leq C h.
$$
\end{Lemma}
\begin{proof}
Since $(1-P_h^*P_h)(H-z)^{-1}=\mathcal{F}^*[1-Q_h^*Q_h](\mathcal{F}H\mathcal{F}^*-z)^{-1}\mathcal{F}$, is enough to prove the result for each entry in the difference matrix of $[1-Q_h^*Q_h](\mathcal{F}H\mathcal{F}^*-z)^{-1}$.  By \cref{eq:resolventefibrado2,coeficientesHh} each entry of $(\mathcal{F}H\mathcal{F}^*-z)^{-1}$ is a linear combination of elements in $\{zR_z(\xi)^{-1}, A_1(\xi)R_z(\xi)^{-1},\dots,A_n(\xi)R_z(\xi)^{-1}\}$. We will just study $A_l(\xi)R_z(\xi)^{-1}$ for a fixed $l$ since the proof for $zR_z(\xi)^{-1}$ is similar. For $\psi\in L^2(\R^n)$ set $g(\xi)=\dfrac{A_l(\xi)}{|2\pi \xi|^2-z}\psi(\xi)$. It is possible to show that (see the Appendix in \cite{NT21})
\begin{align*}
\left(1-Q_h^*Q_h\right)g(\xi)&=g(\xi)-\sum_{\mu\in\{0,\pm1\}^n} \hat{\varphi}(h\xi)\overline{\hat{\varphi}(h\xi+\mu)}g(\xi+h^{-1}\mu)\\
&=(1-|\hat\varphi(h\xi)|^2)g(\xi)-\sum_{0\neq\mu\in\{0,\pm1\}^n} \hat{\varphi}(h\xi)\overline{\hat{\varphi}(h\xi+\mu)}g(\xi+h^{-1}\mu).
\end{align*}
For the first term, by \cref{condi_phi,condi_phi2}
$$
\|(1-|\hat\varphi(h\xi)|^2)g\|\leq \sup_{|\xi|>h^{-1}\delta} \left|\frac{A_l(\xi)}{|2\pi \xi|^2-z}\right| \|\psi\|\leq C h \|\psi\|,
$$
for some $\delta >0$. For the terms in the sum  using \cref{condi_phi,condi_phi2} again we se that $\hat{\varphi}(h\xi)\overline{\hat{\varphi}(h\xi+\mu)}=0$ for $|h\xi+\mu|\leq \delta$.  Thus  for each $\mu\in\{0,\pm1\}^n$
$$
\|\hat{\varphi}(h\xi)\overline{\hat{\varphi}(h\xi+\mu)}g(\xi+h^{-1}\mu)\|\leq \sup_{|\xi+h^{-1}\mu|>\delta} \left|\frac{A_l(\xi+h^{-1}\mu)}{|2\pi (\xi+h^{-1}\mu)|^2-z}\right| \|\psi\|\leq C h \|\psi\|.$$
\end{proof}

\begin{Lemma}\label{le1} For each $z\in\C\setminus \R$, there exists a positive constant $C$ such that 
$$
\|(H_h-z)^{-1}P_h-P_h(H-z)^{-1}\|\leq C h
$$
\end{Lemma}
\begin{proof}
First write 
$$ 
\|(H_h-z)^{-1}P_h-P_h(H-z)^{-1}\|=\|Q_h^* F(H_h-z)^{-1}F^* Q_h-Q_h^*Q_h\F(H-z)^{-1}\F^*\|.  
$$
As is the previous  lemma it is enough to prove that the inequality holds for  each entry in the matrix of $Q_h^* F(H_h-\lambda)^{-1}F^* Q_h- Q_h^*Q_h\F(H-\lambda)^{-1}\F^*$. Then, taking into account \cref{eq:resolventefibrado1,eq:resolventefibrado2}, for  a fixed $l $, we  consider    $\|Q_h^* \dfrac{a_{h,l}}{r_z} Q_h-
Q_h^* Q_h\dfrac{A_l}{R_z} \|$. Let  $\psi$ in  $\S(\R^n)$, then  it is possible to show that (see the Appendix in \cite{NT21})
\begin{equation}\label{NT1}
\left(Q_h^* \dfrac{a_l}{r_z} Q_h -Q_h^*Q_h\dfrac{A_l}{R_z} \right)\psi=\sum_{\mu\in\{0,\pm1\}^n} \hat{\varphi}(h\xi)\overline{\hat{\varphi}(h\xi+\mu)}\mathcal B_h(\xi+h^{-1}\mu)\psi(\xi+h^{-1}\mu),
\end{equation}
where $\mathcal B_h=\dfrac{a_{h,l}}{r_z} -\dfrac{A_l}{R_z}=a_{h,l}\dfrac{R_z-r_z}{r_zR_z}+\dfrac{a_{h,l}-A_l}{R_z}$.
Using the Taylor expansion and \cref{condi_phi,condi_phi2} we easily get  
\begin{equation*}
|\hat\varphi(h\xi)|^2 \left|\frac{a_{h,l}-A_l}{R_z}\right|\leq C h |\hat\varphi(h\xi)|^2,
\end{equation*}
and 
\begin{equation*}
|\hat\varphi(h\xi)|^2|a_{h,l}|\left|\dfrac{R_z-r_z}{r_zR_z}\right|\leq C h |\hat\varphi(h\xi)|^2. 
\end{equation*}
 In the same way, from \cref{condi_phi,condi_phi2},  for  $\mu\in\{0,\pm1\}^n$ 
 \begin{equation}\label{NT2}
|\hat{\varphi}(h\xi)\overline{\hat{\varphi}(h\xi+\mu)}D_h(\xi+h^{-1}\mu)|\leq C h, \quad \xi\in\R^n.
\end{equation}
 Thus, using \cref{NT2} in  \cref{NT1} we finish the proof.

\end{proof}

\subsection{Proof of \texorpdfstring{\cref{teo:main}}{Theorem 1.1}}
Let $T_h:\ell^2(X_h)\to\bigoplus_{j=0}^nL^2(\R^n;\C^{\binom{n}j})$ be the operator
\begin{equation*}
T_h:=P_h^*\tilde{U}_hU\ .
\end{equation*}
 Being $\tilde{U}_hU$ unitary, and $P_h$ a partial isometry with full range, $T_h$ is an isometry.

Noticing that $T_h(d_h+d_h^*-z)T_h^*=P_h^*(H_h-z)P_h$ we write
\begin{equation*}
P_h^*(H_h-z)^{-1}P_h-(H-z)^{-1}
=P_h^*\left[(H_h-\lambda)^{-1}P_h-P_h(H-\lambda)^{-1}\right]-(1-P_h^*P_h) (H-\lambda)^{-1}
\end{equation*}
from where one can conclude by taking into account \cref{le0,le1}.\qedhere

\appendix
\section{Proof of \cref{todofunciona}}\label{apendice}

Because the definition of $\partial$ involves $(\prescript{}{i}{\hat{s}})^*$ we need a better understanding of how the involution and restriction are related before attempting to prove \cref{todofunciona}.{}

\begin{Lemma}\label{stecnico}
Let $\hat{s}$ be given by \cref{shat}. Let $1\leq i_0 \leq j$ and $1\leq i_1 \leq j-1$. Then the following statements hold:
\begin{enumerate}
\item $\prescript{}{i_0}{(\hat{s}^*})=(\prescript{}{j-i_0+1}{\hat{s}})^*$ \label{stecnico1}
\item $(\prescript{}{i_0}{(\hat{s}^*))^*}=(\prescript{}{j-i_0+1}{\hat{s}})$ \label{stecnico2}
\item \begin{equation*}
\prescript{}{i_1}{(\prescript{}{i_0}{\hat{s}})}=\begin{cases}
\prescript{}{(i_0-1)}{(\prescript{}{i_1}{\hat{s}})} &\text{ if }i_1<i_0\ ;\ \\
\prescript{}{i_0}{(\prescript{}{(i_1+1)}{\hat{s}})} &\text{ if }i_0\leq i_1\ .
\end{cases}
\end{equation*}\label{stecnico3}
\end{enumerate}
\end{Lemma}
\begin{proof}
Let $1\leq i\leq j-1$, then on one hand we have 
\begin{align*}
\prescript{}{i_0}{(\hat{s}^*})(i)=&\begin{cases}
\hat{s}^*(i)\quad &i<i_0\\
\hat{s}^*(i+1)\quad & i_0\leq i
\end{cases}\\
=&\begin{cases}
\hat{s}(j-i+1)\quad &i<i_0\\
\hat{s}(j-i)\quad & i_0\leq i
\end{cases}
\end{align*}
while in the other we can compute
\begin{align*}
(\prescript{}{j-i_0+1}{\hat{s}})^*(i)=&\prescript{}{j-i_0+1}{\hat{s}}(j-i)\\
=&\begin{cases}
\hat{s}(j-i)\quad &j-i<j-i_0+1\\
\hat{s}(j-i+1)\quad & j-i_0+1\leq j- i
\end{cases}\\
=&\begin{cases}
\hat{s}(j-i)\quad & i_0\leq i\\
\hat{s}(j-i+1)\quad &i<i_0
\end{cases}
\end{align*}
which proves \cref{stecnico1}. Further, \cref{stecnico2} is a direct consequence of \cref{stecnico1}. To prove \cref{stecnico3} we fix $1\leq i_0\leq j$ and $1\leq i_1\leq j-1$ and first consider the case $i_1<i_0$. We first compute
\begin{align*}
\prescript{}{i_1}{(\prescript{}{i_0}{\hat{s}})}(i)=&\begin{cases}
\prescript{}{i_0}{\hat{s}}(i)\quad &i<i_1\\
\prescript{}{i_0}{\hat{s}}(i+1)\quad & i_1\leq i
\end{cases}\\
=&\begin{cases}
\hat{s}(i)\quad &i<i_1\\
\hat{s}(i+1)\quad &i_1\leq i\leq i_0-2\\
\hat{s}(i+2)\quad &i_0-1\leq i
\end{cases}
\end{align*}
and then compare it with
\begin{align*}
\prescript{}{i_0-1}{(\prescript{}{i_1}{\hat{s}})}(i)
=&\begin{cases}
\prescript{}{i_1}{\hat{s}}(i)\quad &i<i_0-1\\
\prescript{}{i_1}{\hat{s}}(i+1)\quad & i_0-1\leq i
\end{cases}
\end{align*}
which will coincide with the previous expression and hence show that \cref{stecnico3} holds for $i_1<i_0$. The computations for $i_0\leq i_1$ follow the same structure and are omitted.
\end{proof}

\begin{proof}[Proof of \cref{todofunciona}]
Let us first show that $\partial(\overline{s})=\overline{\partial (s)}$. By \cref{sborde} and using \cref{stecnico1,stecnico2} from \cref{stecnico} we can check that:
\begin{align*}
\partial (\overline{s}) =&\cup_{i=1}^j \{(-1)^{j-i} (\lfloor \overline{s} \rfloor;\prescript{}{i}{(\hat{s}^*)}) \}\bigcup\cup_{i=1}^j\{(-1)^{i} (\lceil \overline{s} \rceil;(\prescript{}{i}{(\hat{s}^*)})^*)  \}\\
=&\cup_{i=1}^j \{(-1)^{j-i} (\lceil s \rceil\lfloor ;(\prescript{}{j-i+1}{\hat{s}})^*) \}\bigcup\cup_{i=1}^j\{(-1)^{i} (\lfloor s \rfloor;\prescript{}{j-i+1}{\hat{s}})  \}\\
=&\cup_{m=1}^j \{(-1)^{m-1} (\lceil s \rceil\lfloor ;(\prescript{}{m}{\hat{s}})^*) \}\bigcup\cup_{m=1}^j\{(-1)^{j-m+1} (\lfloor s \rfloor;\prescript{}{m}{\hat{s}})  \}\\
=&\cup_{m=1}^j \overline{\{(-1)^{m} (\lceil s \rceil\lfloor ;(\prescript{}{m}{\hat{s}})^*) \}}\bigcup\cup_{m=1}^j\overline{\{(-1)^{j-m} (\lfloor s \rfloor;\prescript{}{m}{\hat{s}})  \}}=\overline{\partial(s)}\ .
\end{align*}
It remains to show that $\partial(\partial(s))=\overline{\partial(\partial(s))}$. For this we introduce the following notation for $1\leq i \leq j$ and $s\in X^j$
\begin{equation*}
A_i(s):=(-1)^{j-i}(\lfloor s \rfloor ; \prescript{}{i}{\hat{s}})\quad B_i(s):=(-1)^{i}(\lceil s \rceil:(\prescript{}{i}{\hat{s}})^*)\ . 
\end{equation*}
It follows that
\begin{equation*}
\partial(\partial (s))=\cup_{l=1}^{j-1}\cup_{i=1}^{j}\{A_l(A_i(s))\cup B_l(A_i(s))\cup A_l(B_i(s))\cup B_l(B_i(s))\}\ . 
\end{equation*}
Let us now consider $s\in X^j$. Without lose of generality we will assume that $s=(\lfloor s \rfloor ; \delta_1,\dots,\delta_j)$. From the proof of \cref{stecnico} on can see that if $i_1<i_0$, then $\prescript{}{i_1}{(\prescript{}{i_0}{\hat{s}})}$ skips the $i_0$ and $i_1$ terms. We start by identifying the faces $r \in \partial(\partial (s))$ that satisfies $\lfloor r \rfloor = \lfloor s \rfloor$. There are given for $1\leq i\leq j$ and $1\leq l\leq j-1$ such that $j-i$ is even and $j-l$ is odd by
 \begin{equation*}
 A_l(A_i)(s)=A_l((-1)^{j-i}(\lfloor s \rfloor ;  \prescript{}{i}{\hat{s}})=(-1)^{j-1-l}(\lfloor s \rfloor ;  \prescript{}{l}{(\prescript{}{i}{\hat{s}})})=
 \begin{cases}(\lfloor s \rfloor ;  \prescript{}{l}{(\prescript{}{i}{\hat{s}})}) &\text{ if }l<i\\
 (\lfloor s \rfloor ;  \prescript{}{i}{(\prescript{}{l+1}{\hat{s}})}) &\text{ if }i\leq l
 \end{cases}\ ;
 \end{equation*}
and for $1\leq i\leq j$ and $1\leq l\leq j-1$ such that $j-i$ is odd and $l$ is even by
 \begin{align*}
 B_l(A_i)(s)=B_l((-1)^{j-i}(\lfloor s \rfloor ;  \prescript{}{i}{\hat{s}})=&B_l((\lceil s \rceil-\delta_i ;  (\prescript{}{i}{\hat{s}})^*)\\
 =&(-1)^{l}(\lfloor s \rfloor ;  (\prescript{}{l}{((\prescript{}{i}{\hat{s}})^*}))^*)\\
 =&(\lfloor s \rfloor ; \prescript{}{j-l}{(\prescript{}{i}{\hat{s}})})=
 \begin{cases}(\lfloor s \rfloor ;  \prescript{}{j-l}{(\prescript{}{i}{\hat{s}})}) &\text{ if }j-l<i\\
 (\lfloor s \rfloor ;  \prescript{}{i}{(\prescript{}{j-l+1}{\hat{s}})}) &\text{ if }i\leq j-l
 \end{cases}\ .
\end{align*}
In both cases we have faces that start at $\lfloor s \rfloor$ but have two directions less than $s$. To check that we have all possible $\tfrac{j*(j-1)}2$ combinations we first consider the case $j=2k$ for $k\in \mathbb{N}$. We have that the faces $r\in\partial(\partial (s))$  that satisfies $\lfloor r \rfloor=\lfloor s \rfloor$ are given by
\begin{align*}
&\left\{ \cup_{p=1}^k \cup_{m=1}^k A_{2p-1}(A_{2m}(s))\right\}\bigcup\left\{  \cup_{p=1}^{k-1}\cup_{m=1}^k B_{2p}(A_{2m-1}(s))\right\}\\
&\hspace*{30pt}=\left\{ \cup_{p=1}^k \cup_{m=p}^k (\lfloor s \rfloor ;  \prescript{}{2p-1}{(\prescript{}{2m}{\hat{s}})}) \right\}\bigcup\left\{ \cup_{p=1}^k \cup_{m=1}^{p-1} (\lfloor s \rfloor ;  \prescript{}{2m}{(\prescript{}{2p}{\hat{s}})}) \right\}\\
&\hspace*{30pt}\phantom{\quad}\bigcup\left\{ \cup_{p=1}^{k-1}\cup_{m=k-p+1}^{k}(\lfloor s \rfloor ;  \prescript{}{2k-2p}{(\prescript{}{2m-1}{\hat{s}})})\right\}\bigcup\left\{ \cup_{p=1}^{k-1}\cup_{m=1}^{k-p} (\lfloor s \rfloor ;\prescript{}{2m-1}{(\prescript{}{2k-2p+1}{\hat{s}})})\right\}\\
&\hspace*{30pt}=\left\{ \cup_{p=1}^k \cup_{m=p}^k (\lfloor s \rfloor ;  \prescript{}{2p-1}{(\prescript{}{2m}{\hat{s}})}) \right\}\bigcup\left\{ \cup_{p=1}^k \cup_{m=1}^{p-1} (\lfloor s \rfloor ;  \prescript{}{2m}{(\prescript{}{2p}{\hat{s}})}) \right\}\\
&\hspace*{30pt}\phantom{\quad}\bigcup\left\{ \cup_{p=1}^{k-1}\cup_{m=p+1}^{k}(\lfloor s \rfloor ;  \prescript{}{2p}{(\prescript{}{2m-1}{\hat{s}})})\right\}\bigcup\left\{ \cup_{p=1}^{k-1}\cup_{m=1}^{p} (\lfloor s \rfloor ;\prescript{}{2m-1}{(\prescript{}{2p+1}{\hat{s}})})\right\}\\
&\hspace*{30pt}=\left\{ \cup_{l=1}^{j-1} \cup_{i=l+1}^j (\lfloor s \rfloor ;  \prescript{}{l}{(\prescript{}{i}{\hat{s}})}) \right\}.\\
\end{align*}
Similar computations work for $j=2k+1$ as well for the treatment of the other types of faces. We do not enter into the detail of each one but we will discuss where the terms come from with the aid of the following table.

\begin{center}
\begin{tabular}{||l| c c c c||} 
 \hline
  & $r$ & $i$ such that & $l$ such that & $\lfloor r \rfloor$ \\ 
 \hline
 \hline
 1 & $A_l(A_i(s))$ & $j-i$ is even & $j-l$ is odd & $\lfloor s \rfloor$ \\ 
 \hline
 2 & $B_l(A_i(s))$ & $j-i$ is odd & $l$ is even & $\lfloor s \rfloor$ \\
 \hline


 3 & $A_l(A_i(s))$ & $j-i$ is even & $j-l$ is even & $\lceil s \rceil - \delta_{x_1} - \delta_{x_2}$ \\
 \hline
 4 & $B_l(A_i(s))$ & $j-i$ is odd & $l$ is odd & $\lceil s \rceil - \delta_{x_1} - \delta_{x_2}$ \\
 \hline


 5 & $A_l(A_i(s))$ & $j-i$ is odd & $j-l$ is even & $\lfloor s \rfloor + \delta_{x_1}$ \\
 \hline
 6 & $B_l(A_i(s))$ & $j-i$ is even & $l$ is odd & $\lfloor s \rfloor + \delta_{x_1}$ \\
 \hline
 7 & $A_l(B_i(s))$ & $i$ is odd & $j-l$ is odd & $\lfloor s \rfloor + \delta_{x_1}$ \\
 \hline
 8 & $B_l(B_i(s))$ & $i$ is even & $l$ is even & $\lfloor s \rfloor + \delta_{x_1}$ \\
 \hline


 9 & $A_l(A_i(s))$ & $j-i$ is odd & $j-l$ is odd & $\lceil s \rceil  - \delta_{x_1}$ \\
 \hline
 10 & $B_l(A_i(s))$ & $j-i$ is even & $l$ is even & $\lceil s \rceil  - \delta_{x_1}$ \\
 \hline
 11 & $A_l(B_i(s))$ & $i$ is odd & $j-l$ is even & $\lceil s \rceil  - \delta_{x_1}$ \\
 \hline
 12 & $B_l(B_i(s))$ & $i$ is even & $l$ is odd & $\lceil s \rceil  - \delta_{x_1}$ \\
 \hline


 13 & $A_l(B_i(s))$ & $i$ is even & $j-l$ is even & $\lfloor s \rfloor + \delta_{x_1} + \delta_{x_2}$ \\
 \hline
 14 & $B_l(B_i(s))$ & $i$ is odd & $l$ is odd & $\lfloor s \rfloor + \delta_{x_1} + \delta_{x_2}$ \\
 \hline


 15 & $A_l(B_i(s))$ & $i$ is even & $j-l$ is odd & $\lceil s \rceil$ \\
 \hline
 16 & $B_l(B_i(s))$ & $i$ is odd & $l$ is even & $\lceil s \rceil$ \\
 \hline
\end{tabular}
\end{center}
The first two rows corresponds to the calculations we have already done. Rows 3 and 4 will give rise two the opposite faces of the first two rows. In all, the first two rows are all the faces that contain $\lfloor s \rfloor$, either as $\lfloor r \rfloor$ or $\lceil r \rceil$. The next 8 rows describe the faces that contain a vertices at distance $1$ from $\lfloor s \rfloor$. Again we have that the first rows, rows $5$ to $8$, are positively oriented while the second half, rows $9$ to $12$ correspond to theirs opposite faces. Finally, rows $13$ to $16$ are symmetric to the the first 4 rows but its faces contain $\lceil s \rceil$ instead of $\lfloor s \rfloor$. 
\end{proof}

\bigskip
\noindent{\bf Data availability:} There are no data associated with this manuscript.

\noindent{\bf Conflict of interest:} On behalf of all authors, the corresponding author states there is no conflict of interest. 

\noindent{\bf Acknowledgments:} P. Miranda was supported by the Chilean Fondecyt Grant $1201857$. D. Parra was supported by the Chilean Fondecyt Grant $3210686$. 

\printbibliography

@article {NT21,
    AUTHOR = {Nakamura, Shu and Tadano, Yukihide},
     TITLE = {On a continuum limit of discrete {S}chr\"{o}dinger operators on
              square lattice},
   JOURNAL = {J. Spectr. Theory},
  FJOURNAL = {Journal of Spectral Theory},
    VOLUME = {11},
      YEAR = {2021},
    NUMBER = {1},
     PAGES = {355--367},
      ISSN = {1664-039X},
   MRCLASS = {47A10 (47A25 47B39 81U05)},
  MRNUMBER = {4233213},
       DOI = {10.4171/jst/343},
       URL = {https://doi.org/10.4171/jst/343},
}

@article {CGJ21,
    AUTHOR = {Cornean, Horia and Garde, Henrik and Jensen, Arne},
     TITLE = {Norm resolvent convergence of discretized {F}ourier
              multipliers},
   JOURNAL = {J. Fourier Anal. Appl.},
  FJOURNAL = {The Journal of Fourier Analysis and Applications},
    VOLUME = {27},
      YEAR = {2021},
    NUMBER = {4},
     PAGES = {Paper No. 71, 31},
      ISSN = {1069-5869},
   MRCLASS = {47A10 (42B15 47A11 47A58 47B39)},
  MRNUMBER = {4298073},
       DOI = {10.1007/s00041-021-09876-5},
       URL = {https://doi.org/10.1007/s00041-021-09876-5},
}

@Article{Ec45,
  Title                    = {Harmonische {F}unktionen und {R}andwertaufgaben in einem {K}omplex},
  Author                   = {Eckmann, Beno},
  Journal                  = {Comment. Math. Helv.},
  Year                     = {1945},
  Pages                    = {240--255},
  Volume                   = {17},

  Fjournal                 = {Commentarii Mathematici Helvetici},
  ISSN                     = {0010-2571},
  Mrclass                  = {56.0X},
  Mrnumber                 = {0013318},
  Mrreviewer               = {H. Whitney}
}

@Article{Do76,
  Title                    = {Finite-difference approach to the {H}odge theory of harmonic forms},
  Author                   = {Dodziuk, Jozef},
  Journal                  = {Amer. J. Math.},
  Year                     = {1976},
  Number                   = {1},
  Pages                    = {79--104},
  Volume                   = {98},

  Doi                      = {10.2307/2373615},
  Fjournal                 = {American Journal of Mathematics},
  ISSN                     = {0002-9327},
  Mrclass                  = {58A10 (58G99)},
  Mrnumber                 = {0407872},
  Mrreviewer               = {D. B. Fuks},
  Url                      = {http://dx.doi.org.docelec.univ-lyon1.fr/10.2307/2373615}
}

@Article{GW16,
  Title                    = {On eigenvalues of random complexes},
  Author                   = {Gundert, Anna and Wagner, Uli},
  Journal                  = {Israel J. Math.},
  Year                     = {2016},
  Number                   = {2},
  Pages                    = {545--582},
  Volume                   = {216},

  Doi                      = {10.1007/s11856-016-1419-1},
  Fjournal                 = {Israel Journal of Mathematics},
  ISSN                     = {0021-2172},
  Mrclass                  = {55U10 (05C80 05E45 58J50 60D05)},
  Mrnumber                 = {3557457},
  Mrreviewer               = {Dirk Sch\"utz},
  Url                      = {http://dx.doi.org.docelec.univ-lyon1.fr/10.1007/s11856-016-1419-1}
}

@Article{HJ13,
  Title                    = {Spectra of combinatorial {L}aplace operators on simplicial complexes},
  Author                   = {Horak, Danijela and Jost, J{\"u}rgen},
  Journal                  = {Adv. Math.},
  Year                     = {2013},
  Pages                    = {303--336},
  Volume                   = {244},

  Doi                      = {10.1016/j.aim.2013.05.007},
  Fjournal                 = {Advances in Mathematics},
  ISSN                     = {0001-8708},
  Mrclass                  = {55U10 (18G30 31C20)},
  Mrnumber                 = {3077874},
  Mrreviewer               = {Paul G. Goerss},
  Url                      = {http://dx.doi.org/10.1016/j.aim.2013.05.007}
}

@article {Li20,
    AUTHOR = {Lim, Lek-Heng},
     TITLE = {Hodge {L}aplacians on graphs},
   JOURNAL = {SIAM Rev.},
  FJOURNAL = {SIAM Review},
    VOLUME = {62},
      YEAR = {2020},
    NUMBER = {3},
     PAGES = {685--715},
      ISSN = {0036-1445},
   MRCLASS = {58A14 (05C50 20G10)},
  MRNUMBER = {4131346},
       DOI = {10.1137/18M1223101},
       URL = {https://doi.org/10.1137/18M1223101},
}

@article {Ch20,
    AUTHOR = {Chebbi, Yassin},
     TITLE = {Spectral gap of the discrete {L}aplacian on triangulations},
   JOURNAL = {J. Math. Phys.},
  FJOURNAL = {Journal of Mathematical Physics},
    VOLUME = {61},
      YEAR = {2020},
    NUMBER = {10},
     PAGES = {103507, 16},
      ISSN = {0022-2488},
   MRCLASS = {05C50 (31C20 47A10 47B39 58J50)},
  MRNUMBER = {4162922},
       DOI = {10.1063/1.5115778},
       URL = {https://doi.org/10.1063/1.5115778},
}

@article {Ch08,
    AUTHOR = {Chebbi, Yassin},
     TITLE = {The discrete {L}aplacian of a 2-simplicial complex},
   JOURNAL = {Potential Anal.},
  FJOURNAL = {Potential Analysis. An International Journal Devoted to the
              Interactions between Potential Theory, Probability Theory,
              Geometry and Functional Analysis},
    VOLUME = {49},
      YEAR = {2018},
    NUMBER = {2},
     PAGES = {331--358},
      ISSN = {0926-2601},
   MRCLASS = {05C50 (05C10 31C20 39A12 47B25)},
  MRNUMBER = {3824965},
MRREVIEWER = {Aingeru Fern\'{a}ndez Bertolin},
       DOI = {10.1007/s11118-017-9659-1},
       URL = {https://doi.org/10.1007/s11118-017-9659-1},
}

@article {ABDE21,
    AUTHOR = {Athmouni, Nassim and Baloudi, Hatem and Damak, Mondher and
              Ennaceur, Marwa},
     TITLE = {The magnetic discrete {L}aplacian inferred from the
              {G}au\ss -{B}onnet operator and application},
   JOURNAL = {Ann. Funct. Anal.},
  FJOURNAL = {Annals of Functional Analysis},
    VOLUME = {12},
      YEAR = {2021},
    NUMBER = {2},
     PAGES = {Paper No. 33, 30},
      ISSN = {2639-7390},
   MRCLASS = {47B25 (05C12 05C63 81Q35)},
  MRNUMBER = {4228853},
MRREVIEWER = {Veerabhadraiah Lokesha},
       DOI = {10.1007/s43034-021-00119-8},
       URL = {https://doi.org/10.1007/s43034-021-00119-8},
}

@incollection {Ke15,
    AUTHOR = {Keller, Matthias},
     TITLE = {Intrinsic metrics on graphs: a survey},
 BOOKTITLE = {Mathematical technology of networks},
    SERIES = {Springer Proc. Math. Stat.},
    VOLUME = {128},
     PAGES = {81--119},
 PUBLISHER = {Springer, Cham},
      YEAR = {2015},
   MRCLASS = {31C20 (05C12 05C50 47B37 58J50)},
  MRNUMBER = {3375157},
       DOI = {10.1007/978-3-319-16619-3\_7},
       URL = {https://doi.org/10.1007/978-3-319-16619-3_7},
}

@article {ENT22,
    AUTHOR = {Exner, Pavel and Nakamura, Shu and Tadano, Yukihide},
     TITLE = {Continuum limit of the lattice quantum graph {H}amiltonian},
   JOURNAL = {Lett. Math. Phys.},
  FJOURNAL = {Letters in Mathematical Physics},
    VOLUME = {112},
      YEAR = {2022},
    NUMBER = {4},
     PAGES = {Paper No. 83, 15},
      ISSN = {0377-9017},
   MRCLASS = {81Q35 (35J10 47A10 81T27)},
  MRNUMBER = {4470272},
       DOI = {10.1007/s11005-022-01576-5},
       URL = {https://doi.org/10.1007/s11005-022-01576-5},
}

@misc{CGJ22x,
Author = {Horia D. Cornean and Henrik Garde and Arne Jensen},
Title = {Discrete approximations to Dirac operators and norm resolvent convergence},
Year = {2022},
Eprint = {arXiv:2203.07826},
}

@misc{SU21x,
Author = {Karl Michael Schmidt and Tomio Umeda},
Title = {Continuum limits for discrete Dirac operators on 2D square lattices},
Year = {2021},
Eprint = {arXiv:2109.04052},
}

@misc{CGJ22xb,
Author = {Horia Cornean and Henrik Garde and Arne Jensen},
Title = {Discrete approximations to Dirichlet and Neumann Laplacians on a half-space and norm resolvent convergence},
Year = {2022},
Eprint = {arXiv:2211.01974},
}

@misc{PZ22x,
Author = {Olaf Post and Sebastian Zimmer},
Title = {Generalised norm resolvent convergence: comparison of different concepts},
Year = {2022},
Eprint = {arXiv:2202.03234},
}

@book {Po12,
    AUTHOR = {Post, Olaf},
     TITLE = {Spectral analysis on graph-like spaces},
    SERIES = {Lecture Notes in Mathematics},
    VOLUME = {2039},
 PUBLISHER = {Springer, Heidelberg},
      YEAR = {2012},
     PAGES = {xvi+431},
      ISBN = {978-3-642-23839-0},
   MRCLASS = {47A10 (05C90 35J05 35R02 81Q35)},
  MRNUMBER = {2934267},
MRREVIEWER = {Pavel V. Exner},
       DOI = {10.1007/978-3-642-23840-6},
       URL = {https://doi.org/10.1007/978-3-642-23840-6},
}

@article {CT15,
    AUTHOR = {Ann\'{e}, Colette and Torki-Hamza, Nabila},
     TITLE = {The {G}auss-{B}onnet operator of an infinite graph},
   JOURNAL = {Anal. Math. Phys.},
  FJOURNAL = {Analysis and Mathematical Physics},
    VOLUME = {5},
      YEAR = {2015},
    NUMBER = {2},
     PAGES = {137--159},
      ISSN = {1664-2368},
   MRCLASS = {39A70 (05C12 05C50 05C63 47B25)},
  MRNUMBER = {3344097},
       DOI = {10.1007/s13324-014-0090-0},
       URL = {https://doi.org/10.1007/s13324-014-0090-0},
}

@article {Pa17,
    AUTHOR = {Parra, D.},
     TITLE = {Spectral and scattering theory for {G}auss-{B}onnet operators
              on perturbed topological crystals},
   JOURNAL = {J. Math. Anal. Appl.},
  FJOURNAL = {Journal of Mathematical Analysis and Applications},
    VOLUME = {452},
      YEAR = {2017},
    NUMBER = {2},
     PAGES = {792--813},
      ISSN = {0022-247X},
   MRCLASS = {47F05 (05C12 35P05 35P25)},
  MRNUMBER = {3632675},
MRREVIEWER = {Fedor L. Bakharev},
       DOI = {10.1016/j.jmaa.2017.03.002},
       URL = {https://doi.org/10.1016/j.jmaa.2017.03.002},
}

@article {MPR23,
    AUTHOR = {Miranda, Pablo and Parra, Daniel and Raikov, Georgi},
     TITLE = {Spectral asymptotics at thresholds for a {D}irac-type operator
              on {$\mathbb{Z}^2$}},
   JOURNAL = {J. Funct. Anal.},
  FJOURNAL = {Journal of Functional Analysis},
    VOLUME = {284},
      YEAR = {2023},
    NUMBER = {2},
     PAGES = {Paper No. 109743},
      ISSN = {0022-1236},
   MRCLASS = {47A10 (35 47A40 47A55)},
  MRNUMBER = {4507620},
       DOI = {10.1016/j.jfa.2022.109743},
       URL = {https://doi.org/10.1016/j.jfa.2022.109743},
}

@book {Th92,
    AUTHOR = {Thaller, Bernd},
     TITLE = {The {D}irac equation},
    SERIES = {Texts and Monographs in Physics},
 PUBLISHER = {Springer-Verlag, Berlin},
      YEAR = {1992},
     PAGES = {xviii+357},
      ISBN = {3-540-54883-1},
   MRCLASS = {81Q10 (35Q40 47N50 58G25)},
  MRNUMBER = {1219537},
MRREVIEWER = {P. A. Mishnayevskiy},
       DOI = {10.1007/978-3-662-02753-0},
       URL = {https://doi.org/10.1007/978-3-662-02753-0},
}

@article {BMT20,
    AUTHOR = {Bourget, Olivier and Moreno Flores, Gregorio R. and Taarabt,
              Amal},
     TITLE = {One-dimensional discrete {D}irac operators in a decaying
              random potential {I}: {S}pectrum and dynamics},
   JOURNAL = {Math. Phys. Anal. Geom.},
  FJOURNAL = {Mathematical Physics, Analysis and Geometry. An International
              Journal Devoted to the Theory and Applications of Analysis and
              Geometry to Physics},
    VOLUME = {23},
      YEAR = {2020},
    NUMBER = {2},
     PAGES = {Paper No. 20, 51},
      ISSN = {1385-0172},
   MRCLASS = {82B44 (47A10 47B80)},
  MRNUMBER = {4102457},
MRREVIEWER = {Marius Lemm},
       DOI = {10.1007/s11040-020-09341-7},
       URL = {https://doi.org/10.1007/s11040-020-09341-7},
}

@article {CIKS20,
    AUTHOR = {Cassano, B. and Ibrogimov, O. O. and Krej\v{c}i\v{r}\'{\i}k, D. and
              \v{S}tampach, F.},
     TITLE = {Location of eigenvalues of non-self-adjoint discrete {D}irac
              operators},
   JOURNAL = {Ann. Henri Poincar\'{e}},
  FJOURNAL = {Annales Henri Poincar\'{e}. A Journal of Theoretical and
              Mathematical Physics},
    VOLUME = {21},
      YEAR = {2020},
    NUMBER = {7},
     PAGES = {2193--2217},
      ISSN = {1424-0637},
   MRCLASS = {81Q12 (81Q10)},
  MRNUMBER = {4117491},
MRREVIEWER = {Hiroshi T. Ito},
       DOI = {10.1007/s00023-020-00916-2},
       URL = {https://doi.org/10.1007/s00023-020-00916-2},
}

@article {POO21,
    AUTHOR = {Prado, Roberto A. and de Oliveira, C\'{e}sar R. and de Oliveira,
              Edmundo C.},
     TITLE = {Density of states and {L}ifshitz tails for discrete 1{D}
              random {D}irac operators},
   JOURNAL = {Math. Phys. Anal. Geom.},
  FJOURNAL = {Mathematical Physics, Analysis and Geometry. An International
              Journal Devoted to the Theory and Applications of Analysis and
              Geometry to Physics},
    VOLUME = {24},
      YEAR = {2021},
    NUMBER = {3},
     PAGES = {Paper No. 30, 29},
      ISSN = {1385-0172},
   MRCLASS = {47B80 (47A75 81Q10 82B44)},
  MRNUMBER = {4315676},
       DOI = {10.1007/s11040-021-09403-4},
       URL = {https://doi.org/10.1007/s11040-021-09403-4},
}

\end{document}